\documentclass[12pt,reqno]{amsart} 

\usepackage{amssymb} 
\usepackage{amsmath} 
\usepackage{fullpage} 
\usepackage{url}

\newtheorem{theo}{Theorem}
\newtheorem{prop}{Proposition}
\newtheorem{lemma}{Lemma}
\newtheorem{cor}{Corollary}
\theoremstyle{definition}
\newtheorem{defi}{Definition}
\theoremstyle{remark}
\newtheorem{rem}{Remark}
\newtheorem{exa}{Example}

\def\Rb{{\mathbf R}}

\def\Xb{{\mathbf X}}
\def\Yb{{\mathbf Y}}

\def\Ar{{\mathbb A}}

\def\Er{{\mathbb E}}

\def\Hr{{\mathbb H}}

\def\Lr{{\mathbb L}}

\def\Pr{{\mathbb P}}
\def\Qr{{\mathbb Q}}
\def\Rr{{\mathbb R}}

\def\Ac{{\mathcal{A}}}

\def\Fc{{\mathcal{F}}}

\def\Kc{{\mathcal{K}}}

\def\Nc{{\mathcal{N}}}

\def\Sc{{\mathcal{S}}}

\def\Xc{{\mathcal{X}}}

\def\one{{\rm \bf 1}}
\def\({\left(}     
\def\){\right)}    
\def\[{\left[}     
\def\]{\right]}

\def\as{{\frenchspacing a.s.}~}
\numberwithin{equation}{section}

\begin{document}
\title{Group cohesion under individual regulatory constraints} 
\author{Delia Coculescu and Freddy Delbaen}
\address{Institut f\"ur Banking und Finance, Universit\"at Z\"urich, Plattenstrasse
   14, 8032 Z\"{u}rich, Switzerland}
\address{Departement f\"ur Mathematik, ETH Z\"urich, R\"{a}mistrasse
   101, 8092 Z\"{u}rich, Switzerland}
   \address{ 
   Institut f\"ur Mathematik,
 Universit\"at Z\"urich, Winterthurerstrasse 190,
 8057 Z\"urich, Switzerland}
\date{\today}
\thanks{The first author thanks the participants at the seminar Finance and Insurance at the University of Zurich, in particular Pablo Koch Medina for helpful interactions. }
\maketitle

\begin{abstract}
We consider a group consisting of $N$ business units.  We suppose there are regulatory constraints for each unit, more precisely,  the net worth of each business unit is required to belong to a set of acceptable risks, assumed to be a convex cone. Because of these requirements, there are less incentives to operate under a group structure, as creating one single business unit, or altering the liability repartition among units, may allow to reduce the required capital.  We analyse the possibilities for the group to benefit from a diversification effect and economise on the cost of capital.  We define and study the risk measures that allow for any group to achieve the minimal capital, as if it were a single unit, without altering the liability of business units, and despite the individual admissibility constraints. We call these risk measures cohesive risk measures.
\end{abstract} 
\section{Introduction}
We consider an insurance group structured in $N\geq 2$  business units. Each unit $i$ has some exogenous liability, modelled as a random variable $X_i\geq 0$ on a probability space $(\Omega,\Fc,\Pr)$. We suppose that the net worth of each unit is subject to constraints, e.g. that are set by a regulator,  namely it needs  to belong to  a certain set  $\Ac$ of ``acceptable positions''. We consider that $\Ac$ is a convex cone so that  the functional $\rho:L^1(\Omega,\Fc,\Pr)\to \Rb$:
$$
\rho(\xi)=\inf\{m\;:\; \xi+m\in\Ac\}
$$is a coherent risk measure (see \cite{ADEH2}). 

Whenever the aggregation of the units' liabilities is possible,  the group is only required  to hold the capital  $\rho\(-\sum_{i=1}^NX_i\)$.  By convexity of $\rho$ we have that $\rho\(-\sum_{i=1}^NX_i\)\leq \sum_{i=1}^N\rho\(-X_i\)$, which reflects the fact that the group would achieve a lower required capital as compared with $N$ separated entities with the same liabilities.  In this situation, risk aggregation is beneficial because it reduces the capital.

In this paper we are assuming  that there are legal or geographical limitations that  prevent  risk transfers or aggregation of liability to take place with the aim of reducing the regulatory capital. Each business unit must face the regulatory requirements individually. At the group level nevertheless, it is possible to manage the available capital and make certain monetary transfers to compensate for losses occurring at time 1 within the different  business units. The set of all such possible monetary compensations  will be called admissible payoffs. 
For instance, when at the group level the available capital is $m>0$, then the set of admissible payoffs will be denoted by $\Ar^\Xb(m)$; it contains the nonnegative, $N$ dimensional random vectors $\Yb=(Y_1,...,Y_N)$ such that $\sum_iY_i=m$, and fulfilling some additional rules specifying  the payment priority of the different units (as given in Definition \ref{defi} below).

In this framework, unit $i$ receives a payoff $Y_i$ so that its net worth is $Y_i-X_i$. The lowest overall capital that the group needs to hold when it has liability given by the vector $\Xb=(X_1,...,X_N)$ is:
\begin{equation}\label{ka}
\Kc(\Xb):=\inf\{m\geq 0\;|\; \exists \Yb\in\Ar^\Xb(m),\forall i, Y_i-X_i\in\Ac\}. 
\end{equation}
In general, when the business units are facing such individual admissibility constraints, it is the case that the cost of capital $\Kc(\Xb)$ is higher than the minimal cost obtained with aggregating the risks, $\rho\(-\sum_{i=1}^NX_i\)$.  Hence, the existence of individual constraints reduces the benefit of being a group. In such circumstances the incentives are to organise the business differently, as a unique entity, or form some optimised subgroups, depending on the liability vector.  

The topic of this paper, is to characterise the acceptability sets $\Ac$ that satisfy the property
\begin{equation}\label{eq}
\Kc(\Xb)=\rho\(-\sum_{i=1}^NX_i\),\;\forall \Xb\in (L^\infty)^N.
\end{equation}
We will show that the relation (\ref{eq}) is rather restrictive. 
We call the corresponding  risk measures cohesive, as any group requires the same amount of capital as if it were a single entity,  despite the impossibility to aggregate liability to take advantage of the convexity of the regulatory constraint. When the risk measure is cohesive, the group benefits of the maximal diversification gain, as if it were a single entity, even with individual  capital constraints for the group members.

On the way of  characterising the cohesive risk measures, we will show how admissible payoffs can be designed in order to offset the liability for each business unit and achieve an acceptable net worth. 
%\begin{theo}
%Suppose that $\rho$ satisfies (\ref{eq}). Then, there exists a probability measure $\Hr\ll\Pr$ so that $\rho=ES$. 
%\end{theo}
%
%This proves that the expected shortfall risk measure is the only risk measure that  preserves the benefit of diversification entirely, under individual admissibility constraints. 

Our problem formulation is connected to the topic of optimal risk transfers based on convex risk measures. Optimal risk transfers within a group is a topic that has been studied in a substantial body of literature. We refer to Heath and Ku  \cite{HeaKu04},  Barrieu and El Karoui \cite{BarrElK05}, \cite{BarrElK05a}, Jouini et al. \cite{JouSchTou08},  Filipovi\'c and Kupper \cite{FilKup08}, Burgert and R\"uschendorf \cite{BurRus06}, Embrechts et al. \cite{{EmbLiuMao}}. In these papers, the problem is formulated generally as
$$
\inf_{\boldsymbol \xi}\sum_{i=1}^N\rho_i(-\xi_i)
$$
over all vectors $\boldsymbol \xi=(\xi_1,\cdots,\xi_N)$ satisfying $\sum_{i=1}^N\xi_i=\eta$ for some $\eta$ fixed. Note that each business unit may use a specific risk measure in this framework.
 The main difference with the optimal risk transfer literature is to introduce individual risk admissibility constraints for every unit.  At the same time,  we consider that the liability at the level of each unit is not transferable among units and we introduce solvability constraints, namely that payments can only be made within the limits of the available capital. Further,  when the group is insolvent, we introduce fixed rules for how the payments are to be made. This framework is similar to the one we have introduced in \cite{CocDel20}, where the question addressed was the fairness  of insurance contracts in presence of default risk.  With similar rules for payments in bankruptcy and admissibility conditions for the payments, we have shown that it is not possible in general to perfectly offset the default risk exposure of the insured agents by proposing them a benefit participation. The question of when such offsetting payments can take place was not addressed there and the current paper also brings clarifications in that context.

\section{Setup and main definitions}
Let us  introduce the mathematical setup more clearly. We work in a two date  model: time 0 where everything is known and time 1, where randomness is present. Possible outcomes at time 1 are modelled as random variables on a probability space $(\Omega, \Fc, \Pr)$, considered to be atomless.    Unless otherwise specified, all equalities and inequalities involving random variables are to be considered in an $\Pr \;\as$ sense.
 The space of risks occurring at time 1 is considered to be $L^\infty(\Omega, \Fc, \Pr)$, simply denoted $L^\infty$, i.e.,  the collection of all essentially bounded random variables.

At time 0,  a regulator measures risks by means of a convex functional $\rho$ fulfilling the properties detailed below.
 \begin{defi} A mapping $\rho\colon L^\infty\rightarrow \Rr$ is called a coherent risk measure if the following properties hold:
 \begin{enumerate}
 \item if $ \xi\geq 0$ then $\rho(\xi)\le 0$;
 \item $\rho$ is convex: for all $\xi,\eta\in L^\infty$, $0\le \lambda\le 1$ we have $\rho(\lambda \xi +(1-\lambda)\eta)\le \lambda \rho(\xi) + (1-\lambda) \rho(\eta)$;
 \item for $a\in \Rr$ and $\xi\in L^\infty$, $\rho(\xi +a)=\rho(\xi) - a$;
 \item  for all $0\le \lambda\in \Rr$, $\rho(\lambda \xi)=\lambda \rho(\xi)$;
 \item (the Fatou property) for any sequence  $\xi_n\downarrow\xi$ (with $\xi_n\in L^\infty$) we have $\rho(\xi_n)\uparrow \rho(\xi)$.
 \end{enumerate}
 \end{defi}
 
 We refer to \cite{ADEH2}, \cite{Pisa}, \cite{FDbook} for an interpretation of these mathematical properties and how they apply to the framework of risk regulation. The main idea is that the regulator only accepts risks $\xi$  that satisfy $\rho(\xi)\leq 0$, hence we say that a random variable $\xi$ is acceptable whenever $\rho(\xi)\le 0$.  Remark that $\xi+\rho(\xi)$ is always acceptable, so that $\rho(\xi)$ is interpreted as the capital required for the risk $\xi$.  If $\rho$ is coherent then the acceptability set 
$$
\Ac:=\{\xi\mid \rho(\xi)\le 0\}
$$ 
is a convex cone.  

The Fatou property allows to apply convex duality theory and establishes a one-to-one correspondence between a coherent risk measure  and a convex closed set $\Sc$ consisting of probabilities which are absolutely continuous with respect to $\Pr$ (the so called  scenario set of $\rho$):
\begin{theo} If $\rho$ is coherent, there exists a convex closed set $\Sc\subset L^1$, consisting of probability measures, absolutely continuous with respect to $\Pr$, such that for all $\xi \in L^\infty$:
$$ \rho(\xi)=\sup_{\Qr\in\Sc} \Er_\Qr[-\xi].$$
Conversely each such a set $\Sc$ defines a coherent utility function.
\end{theo}

We shall use the assumption that $\Sc$ is weakly compact, so that we will be able to replace the sup by a max. Indeed, as a direct application of James's theorem, weak compactness is equivalent to the nonemptyness of the subgradient of $\rho$ at any point: for every $\xi\in L^\infty$, $\nabla \rho(\xi)\neq \emptyset$, that is, there is a $\Qr\in\Sc$ with $\rho(\xi)=\Er_\Qr[-\xi]$.

An additional assumption that we will use is that the risk measure $\rho$ used by the regulator is commonotonic.
  
\begin{defi}  We say that two random variables $\xi,\eta$ are commonotonic if there exist a random variable $\zeta$ as well as two non-decreasing functions $f,g\colon \Rr\rightarrow\Rr$ such that $\xi=f(\zeta)$ and $\eta=g(\zeta)$.
\end{defi}

\begin{defi} We say that $\rho\colon L^\infty\rightarrow\Rr$ is commonotonic if for each couple $(\xi,\eta)$ of commonotonic random variables we have $\rho(\xi+\eta)=\rho(\xi)+\rho(\eta)$.
\end{defi}

\begin{rem}  Loosely speaking, two  risks are commonotonic if they are bets on the same event. Indeed, $\xi$ and $\eta$ being nondecreasing functions of $\zeta$, neither of them is a hedge against the other. The commonotonicity of $\rho$ can therefore be seen as a translation of the rule:  if there is no diversification, there is also no gain in putting these claims together.
\end{rem}
\begin{rem}  If $\rho$ is commonotonic then for nonnegative random variables $f,g\in L^\infty$ satisfying $\Pr[f>0,g>0]=0$ we have that $\rho(f+g)=\rho(f)+\rho(g)$. In particular for $\xi\in L^\infty$: $\rho(\xi)=\rho(\xi^+)+\rho(-\xi^-)$. We also have that for $\Qr\in\Sc$ satisfying $\rho(\xi)=\Er_\Qr[-\xi]$, necessarily also $\rho(\xi^+)=\Er_\Qr[-\xi^+]$ and $\rho(-\xi^-)=\Er_\Qr[\xi^-]$. This easily follows  from the subadditivity of $\rho$.
\end{rem}

All notions above are standard in the theory of risk measures. We now introduce some definitions that are specific to the framework of this paper, that is the one  of a group consisting of $N$ distinct units under regulatory supervision.

 \begin{defi}\label{defi} We denote by $\Xc$ the space of $N$ dimensional random variables which are positive and bounded. 
 We consider a liability vector $\Xb=(X_1,\cdots,X_N)\in \Xc$  and consider $m\in\Rb_+$. The class of admissible payoffs from a total capital $m$,  corresponding to the liability $\Xb$ is defined as:
 $$
 \Ar^\Xb(m)=
 \left \{ \Yb\in \Xc \biggm | \sum_{i=1}^NY_i=m;\; \forall k\in\{1,...,N\}:
 \begin{array}{l}
 \text{ if  $\sum_{i=1}^N X_i>m$ then $Y_k=\frac{X_k}{\sum_{i=1}^NX_i} m$} \\
  \text{ if  $\sum_{i=1}^N X_i\leq m$ then $Y_k\geq X_k$} 
  \end{array}
 \right  \}.
 $$
 \end{defi}
 
  Admissible payoffs respect some rules as follows. If the capital $m$ is less than the aggregate liability $\sum_i X_i$, the group defaults. In this case, all liabilities have the same priority of payment, regardless the unit to which they are corresponding.  Hence, in default, all capital $m$  is distributed towards the units proportionally to their liability size.  Whenever the group is solvent at the aggregate level ($\sum_i X_i\leq m$), every unit must be solvent as well, hence the central unit distributes for each unit $i$  a payment that should cover the liability $X_i$. As there is a surplus in this case, some units  will get more than their liability as a payoff. 
    
\begin{exa} Let us consider that each business unit receives some constant proportion of the surplus $(m-\sum_i X_i)^+$.  The corresponding admissible payoffs  (that we shall call standard payoffs) are given as follows:
 \begin{align}\label{formYi}
 Y_k&=\[X_k+\alpha_k\(m-\sum_i X_i\)\]\one_{\{\sum_i X_i\leq m\}}+X_k\(\frac{m}{\sum_i X_i}\)\one_{\{\sum_i X_i> m\}},\quad i=1,\ldots,N
 %\\\label{formY0}
% Y_0&=\[k_0 +\alpha_0\(k-S^\Xb\)\]\one_{\{S^\Xb\leq k\}}+\(K-S^\Xb\)^+\one_{\{S^\Xb> k\}},
\end{align}
where  each $\alpha_k$ is a nonnegative constant  and  $\sum_{i=1}^N\alpha_i=1$. 
 \end{exa}

\begin{defi}\label{defiOptPay}
Given a level of capital $m$,  an  \textit{offsetting payoff} corresponding to the liability $\Xb\in\Xc$, is a vector of random variables $\Yb\in\Ar^\Xb(m)$ satisfying 
\begin{equation}
Y_i-X_i\in \Ac \text{ for all $i\in\{1,...,N\}$.}
\end{equation} that is, the net worth of any unit is acceptable.
\end{defi}
Offsetting payoffs cannot be achieved when there is not sufficient overall capital $m$.   The analysis in the next section will reveal the fact that the offsetting payoffs can never be achieved  when the capital $m$ is less than $K:=\rho(-\sum_i X_i)$, that is the minimal capital required for the aggregated liability. Also, in general, holding  the capital $K$ does not guarantee the existence of these payoffs, so that the group may be required to hold more capital. This justifies to introduce the following additional  definition: 

\begin{defi}
A coherent risk measure $\rho$ is called \textit{cohesive}  if for any risk vector $\Xb\in\Xc$ there exists $\Yb\in\Ar^\Xb(\rho(-\sum_{i=1}^N X_i))$ such that
\begin{equation}
Y_i-X_i\in \Ac \text{ for all $i\in\{1,...,N\}$.}
\end{equation}
\end{defi}

\section{Properties of  offsetting payoffs with minimal capital}\label{Secoffsetting}
We use the setup and notation from the previous section; in particular $\rho$ is a coherent risk measure that is commonotonic, with a corresponding scenario set $\Sc$ assumed to be weakly compact. Also, we shall consider a fixed liability vector $\Xb\in\Xc$  and denote the aggregated group liability by
$$
S^\Xb:=\sum_{i=1}^N X_i.
$$
Also we denote
$$
K:=\rho\(-\sum_{i=1}^N X_i\)
$$
that is the minimum capital for the aggregated liability. Under individual regulatory constraints for the business units, the minimum regulatory capital for the group is denoted $\Kc(\Xb)$ and its expression was introduced  in (\ref{ka}). 

 First observation is that the group needs to hold at least a capital of $K$. 
\begin{lemma}\label{lemi}
The minimal capital for the group $\Kc(\Xb)$ satisfies  $\Kc(\Xb)\geq K$.
\end{lemma}
\begin{proof}
We  assume $\Yb\in \Ar^\Xb(m)$ offsetting, that is $Y_i-X_i\in \mathcal A$, or $\rho(Y_i-X_i)\leq 0$ for all $i\in\Nc$. From this  and the sub-additivity of $\rho$, we get:
\begin{equation}\label{eqi}
0\geq  \sum_{i=1}^N \rho(Y_i-X_i)\geq \rho\left(\sum_{i=0 }^N (Y_i-X_i)\right)= \rho\(m-\sum_{i=0}^NX_i\)=K-m. 
\end{equation}

\end{proof}
We now investigate what happens if the company holds a capital of $K$.
Is it possible to split the capital  $K$ in a vector of admissible payoffs $\Yb$, so that the net worth $Y_i-X_i$ of each unit $i$ is acceptable? The existence of such payoffs is not granted. Below we show that this condition is rather restrictive and we characterise the situations where the answer to the question is positive.  

A first remark is that the existence of offsetting payoffs with capital $K$ requires that no further improvement of the net worth of the business units can be reached by further diversification. This is what the next lemma says.
\begin{lemma}\label{lemii}
If the  payoff vector $\Yb\in \Ar^\Xb(K)$ is offsetting the liability $\Xb$, then:
\begin{equation}\label{eqii}
\sum_{i=0}^N\rho(Y_i-X_i)=\rho\(\sum_{i=0}^N (Y_i-X_i)\).
\end{equation}
\end{lemma}

\begin{proof} The proof is similar to the one of Lemma \ref{lemi}. It suffices to take $m=K$ in (\ref{eqi}), therefore we need to have only equalities. We remark that the admissibility of $\Yb$ does play a role in establishing this result only through the condition $\sum_i Y_i=K$
\end{proof}

\begin{prop}\label{propequiv}Consider a payoff vector $\Yb\in \Ar^\Xb(K)$. The following are equivalent:
\begin{itemize}
\item[(i)] The  payoff vector $\Yb$ is offsetting the liability $\Xb$.
\item[(ii)] If  $\Qr^*\in\nabla S^\Xb$, then for all $i\in\{1,...,N\}$: 
\begin{align}\label{uY-X}
\rho(Y_i-X_i)& =\Er_{\Qr^*} [X_i-Y_i]=0.
\end{align}
\item[(iii)] Relation (\ref{eqii}) holds and  for all $i\in\{1,...,N\}$: 
\begin{align}\label{uY-Xbis}
\Er_{\Qr^*} [X_i-Y_i]=0.
\end{align}
for some    $\Qr^*\in\nabla S^\Xb$.

\item[(iv)]The following hold: the minimal group capital $\Kc(\Xb)$ satisfies
\begin{align*}
K=\Kc(\Xb)
\end{align*}
and $\Yb$ is a solution of 
$$
\inf_{\xi\in \Ar^\Xb(K)}\sum_{i=0}^N \rho(\xi_i-X_i).
$$
\end{itemize}
\end{prop}
\begin{proof}
(i)$\Rightarrow $(ii). We can apply Proposition \ref{proplin} below, taking $E=\{1,\ldots,N\}$ and  $\xi_i:=Y_i-X_i$. Indeed, by Lemma \ref{lemii}, the condition in Proposition \ref{proplin} (1) is fulfilled and it is equivalent to (3). \\
(ii)$\Rightarrow$(i). Obvious.\\
(ii)$\Rightarrow$(iii). The relation (\ref{uY-X}) implies that relation (\ref{eqii}) holds as it can easily checked; 
also (\ref{uY-X}) implies (\ref{uY-Xbis}) obviously.\\
(iii)$\Rightarrow$(ii). If relation (\ref{eqii}) holds, we can apply Proposition \ref{proplin} below to deduce that  $\rho(Y_i-X_i) =\Er_{\Qr^*} [X_i-Y_i]$, for all   $\Qr^*\in\nabla S^\Xb$ and these expressions are null, again by (iii). \\
(i)$\Leftrightarrow$(iv). In general, $\sum_{i=0}^N \rho(\xi_i-X_i)\geq  \rho\(\sum_{i=0}^N(\xi_i-X_i)\)$ so that
$$
\inf_{\xi\in \Ar^\Xb(K)}\sum_{i=0}^N \rho(\xi_i-X_i)\geq \rho(K-S^\Xb)=0.
$$
The condition $K=\Kc(\Xb)$ means that there are offsetting payoff vectors with a capital $K$, while from the above inequality we see that whenever there exist such offsetting payoff vectors, they are solving the minimisation problem (as $\rho(Y_i-X_i)=0$ whenever $\Yb$ offsetting $\Xb$). Hence the proof is complete.
\end{proof}

\begin{prop}\label{proplin}
We consider some random variables $(\xi_i)_{i\in E}$, with $E$ some countable set and let $\Qr^E\in \nabla{\rho\(\sum_{i\in E} \xi_i\)}$. 
The following are equivalent: 
\begin{itemize}

\item[(1)] 
\begin{equation}\label{xilin0}
\rho\left (\sum_{i\in E}\xi_i\right)= \sum_{i\in E}\rho(\xi_i).
\end{equation} 
\item[(2)] For all $ \lambda_i\geq 0$, $i\in E$
\begin{align}\label{xilin}
\rho\left (\sum_{i\in E} \lambda_i\xi_i\right)&=\sum_{i\in E}\lambda_i \rho\left (\xi_i\right).
\end{align}
\item[(3)] For all  $i\in E$
\begin{align}\label{xilin2}
\rho\left (\xi_i\right)&=\Er_{\Qr^E}\left [-\xi_i\right].
\end{align}
\end{itemize}
\end{prop}

\begin{proof}
We show $(1)\Rightarrow (3)$.
For any  $i\in E$, let  $\Qr^{\{i\}}\in \Sc$ be such that  $\rho\(\xi_i\)=\Er_{\Qr^{\{i\}}}\[-\xi_i\].$ Then:
$$
\sum_{i\in E}\rho(\xi_i)=\sum_{i\in E}\Er_{\Qr^{\{i\}}}\[-\xi_i\]\geq \sum_{i\in E}\Er_{\Qr^E}\[-\xi_i\]=\rho\(\sum_{i\in E} \xi_i\).
$$
From (\ref{xilin0}), it follows that we must have only equalities above, hence:
$$
\sum_{i\in E}\left (\Er_{\Qr^{\{i\}}}[-\xi_i]- \Er_{\Qr^E}[-\xi_i]\right )=0,
$$
which implies (as all terms in the above sum are nonnegative) that  $\rho\(\xi_i\)=\Er_{\Qr^{\{i\}}}[-\xi_i] = \Er_{\Qr^E}[-\xi_i]$. 

Now, we show $(3)\Rightarrow (2)$.
Using the linearity of the expectation and (\ref{xilin2}) we obtain:
\begin{equation}\label{eq1}
\rho\left (\sum_{i\in E} \lambda_i\xi_i\right)\geq \Er_{ \Qr^E}\left [-\sum_{i\in E} \lambda_i\xi_i\right]=\sum_{i\in S} \lambda_i\rho(\xi_i).
\end{equation}
On the other hand, $\rho$ being  convex  we also have for all $\lambda_i \geq 0$:
\begin{equation}\label{eq2}
\rho\left (\sum_{i\in E} \lambda_i\xi_i\right)\leq \sum_{i\in E} \lambda_i\rho(\xi_i)
\end{equation}
Combining (\ref{eq1}) and (\ref{eq2}) we get the equality (\ref{xilin}). 

The implication $(2)\Rightarrow (1)$ is trivial, hence the proof is complete.

\end{proof}

Proposition \ref{propequiv} (ii) emphasizes that a payoff vector $\Yb$ satisfying  $\Er_{\Qr^*}[X_i-Y_i]=0$ with  $\Qr^*\in\nabla \rho(-S^\Xb)$ is a potential candidate to be an offsetting payoff (necessary condition). It follows that  when payoffs are standard, i.e., as in (\ref{formYi}),  the proportions $(\alpha_i)$ can be identified via these equalities. Their expressions are given below in Proposition \ref{offsettingStandardC},  together with another necessary and sufficient condition for these to be indeed offsetting vectors. To be noted that this time we will use commonotonicity of the risk measure $\rho$ to obtain this condition, the previous results remaining true also when $\rho$ not commonotonic.

\begin{prop}\label{offsettingStandardC} The minimal capital for the group $\Kc(\Xb)$ satisfies  $\Kc(\Xb)= K$ if and only if
\begin{equation}\label{cnsoff}
\rho\(-\(K-S^\Xb\)^-\)=\sum_{i=1}^N\rho\(-\frac{X_i}{S^\Xb}(K-S^\Xb)^-\).
\end{equation}
If this condition is satisfied, then $\Yb$ is offsetting $\Xb$, where $\Yb$  is a  standard payoff (see  (\ref{formYi})) with
\begin{align*}
\forall i:\quad\alpha_i:&=\frac{\Er_{\Qr^*}\[ \frac{X_i}{S^\Xb} (K-S^\Xb)^-\]}{\Er_{\Qr^*}\[ (K-S^\Xb)^+\]}
=\frac{\Er_{\Qr^*}\[ \frac{X_i}{S^\Xb} (S^\Xb-K)^+\]}{\Er_{\Qr^*}\[ (S^\Xb-K)^+\]}
\end{align*}
where $\Qr^*\in\nabla \rho(-S^\Xb)$. 
% If there exist  payoffs in $\Ar^\Xb(K)$ offsetting $\Xb$,  then $\Yb$ is offsetting $\Xb$.
\end{prop}
\begin{proof}
By Proposition \ref{propequiv} (iii),  $\widetilde \Yb\in \Ar^\Xb(K)$ is ofsetting $\Xb$ if and only if the condition (\ref{eqii}) is satisfied together with $\Er_{\Qr^*}[\widetilde Y_i-X_i]=0$ for all $i$.

Using the commonotonicity of $\rho$  we obtain  that an equivalent expression for (\ref{eqii}) is (applied to $\widetilde \Yb$):
\begin{align*}
& \[\rho\(\(\sum_{i=1}^N\widetilde  Y_i-X_i\)^+\)-\sum_{i=1}^N\rho\((\widetilde Y_i-X_i)^+\)\]\\
 &+\[\rho\(-\(\sum_{i=1}^N \widetilde Y_i-X_i\)^-\)-\sum_{i=1}^N\rho\(-(\widetilde Y_i-X_i)^-\)\]=0
\end{align*}
and because of the subadditivity property of $\rho$ each of the two expressions in the brackets is smaller or equal to 0. It follows the equivalent expression of (\ref{eqii}) is:
\begin{equation}\label{cns1}
\rho\(\(\sum_{i=1}^N\widetilde  Y_i-X_i\)^+\)=\sum_{i=1}^N\rho\((\widetilde Y_i-X_i)^+\)
\end{equation}
and
\begin{equation}\label{cns2}
\rho\(-\(\sum_{i=1}^N \widetilde Y_i-X_i\)^-\)=\sum_{i=1}^N\rho\(-(\widetilde Y_i-X_i)^-\).
\end{equation}
We notice that because $\widetilde \Yb$ is admissible, the expression (\ref{cns2}) equals (\ref{cnsoff}), so that the condition (\ref{cnsoff}) is necessary for the existence of ofsetting payoffs. We now show that it a sufficient condition. Indeed, we can always choose $\widetilde \Yb=\Yb$ with $\Yb$ standard and as stated in the proposition, we have that the random variables  $( Y_i-X_i)^+=\alpha_i (K-S^\Xb)^+$, $\forall i=1,...,N$. We then use the commonotonicity property of $\rho$ to conclude that (\ref{cns1}) is verified. The particular proportions $\alpha_i$ are found through the equality  $\Er_{\Qr^*}[Y_i-X_i]=0$ Proposition \ref{propequiv} (ii). Hence, once (\ref{cnsoff}) is verified,   the vector $\Yb$ fulfils the necessary and sufficient conditions to be offsetting.
\end{proof}

\section{A class of cohesive risk measures}

We consider  a random variable $H\geq 0$ that satisfies $1\leq \Er_{\Pr}[H]<\infty$ and introduce the risk measure: 
\begin{equation}\label{properu}
\rho^H(\xi):=\sup_{\Qr\in\Sc^H}\Er_{\Qr}[-\xi]
\end{equation} with a scenario set: 
\begin{equation}\label{propers}
\Sc^H:=\left \{\Qr\mid 0\leq \frac{d\Qr}{d\Pr}\leq H \quad \Pr\; a.s.\right \}
\end{equation}

The main result in this subsection is  that all  cohesive risk measures   that are commonotonic have this representation. A generalisation of these risk measures will follow afterwards.  Before we prove this result, we give some alternative characterisations of  $\rho^H$:

\begin{prop}\label{uH=AVar}
\begin{itemize}
\item[(1)]Let us denote $\Er_{\Pr}[H]=h$ and introduce the probability measure $\Hr \ll \Pr$ as: 
 $$
 \frac{d\Hr}{d\Pr}:=\frac{H}{h}.
 $$
For any $\xi\in L^\infty$ we have:
 \begin{align*}
\rho^H(\xi)&=AV@R_{(\Hr,1/h)}(\xi),
\end{align*}i.e.,  the average value at risk for $\xi$, at the level $1/h$ and under the probability $\Hr$. We recall  $AV@R_{(\Pr,\lambda)}(\xi)$ is defined as:
$$
AV@R_{(\Pr,\lambda)}(\xi)=\max_{\Qr\in\Sc_\lambda}\Er_{\Qr}[-\xi],
$$where $\Sc_\lambda$ is the set of all probability measures $\Qr\ll\Pr$ whose density $d\Qr/d\Pr$ is $\Pr$ \as bounded by $1/\lambda$.
\item[(2)]For any $\xi\in L^\infty$ we have:
 \begin{equation*}
\rho^H(\xi)=\Er_{\Qr^\xi}[-\xi]
\end{equation*} for a probability $\Qr^\xi$  satisfying:
\begin{equation}\label{Q0}
\frac{d\Qr^\xi}{d\Pr}=
\begin{cases} 
H &\text{ on } \{\xi<q \}\\
\frac{c}{h}H &\text{ on } \{\xi=q \}\\
0 &\text{ on } \{\xi > q\},
\end{cases}
\end{equation}
with $q=\inf\{x\;:\; \Er_\Pr\[H\one_{\xi\leq x}\]\geq 1\}$
 and:
\begin{equation*}
c=
\begin{cases} 
0 &\text{ if }\text{  } \Er_\Pr[H\one_{\xi=q}]=0\\
\frac{1-\Er_\Pr[H\one_{\xi<q}]}{ \Er_\Pr[H\one_{\xi=q}]} &\text{ otherwise. }
\end{cases}
\end{equation*}

\end{itemize}
\end{prop}

\begin{proof}
\begin{align*}
\rho^H(\xi)&=\sup \left\{-\Er_\Pr[\varphi \xi]\;|\; 0\leq \varphi\leq H\;,\; \Er_\Pr[\varphi]=1\right\} \\
&=\sup\left\{-\Er_{\Hr}[\psi \xi ]\;|\; 0\leq \psi\leq h\;,\; \Er_{\Hr}[\psi]=1\right\}\\
&=AV@R_{(\Hr,1/h)}(\xi).
\end{align*}
and we proved (1). 
By applying the Neyman-Pearson lemma, one can show (2). For more details, see also next section, where a generalisation appears, or Subsection 4.4 in \cite{FDbook}.
\end{proof}
Because $AV@R$ is commonotonic it follows that
\begin{cor}\label{corcomp}
The risk measure $\rho^H$ is commonotonic.
\end{cor}
 
We are now ready to prove the main result of this section.

\begin{prop}\label{PropuH} A  commonotonic risk measure $\rho$ satisfying the weak compactness property is  cohesive if and only if it has the representation (\ref{properu}) for some random variable $H\in L^1(\Pr)$.
\end{prop}

\begin{proof}  We first proof that $\rho^H$ is cohesive. The risk measure $\rho^H$ is commonotonic (Corrollary \ref{corcomp}) and clearly satisfies the weak compactness property. It is cohesive if for any given liability vector $\Xb$ offsetting payoffs exist. For  $\Xb\in\Xc$ and $S^\Xb=\sum_iX_i$,  let $\Qr^*$ be such that $K=\rho^H(-S^\Xb)=\Er_{\Qr^*}[S^\Xb]$. It is sufficient to show that  for arbitrary $\Xb\in\Xc$ and for $\Yb$ as in Proposition \ref{offsettingStandardC}, the conditions $Y_i-X_i\in \Ac$, $\forall\;i\in\{1,...,N\}$ are satisfied.   By construction the condition $\Er_{\Qr^*}(X_i-Y_i)=0$ is satisfied for all $i\geq 1$ and it remains to prove: 
\begin{equation}\label{uH}
\rho^H(Y_i-X_i)=\Er_{\Qr^*}(X_i-Y_i).
\end{equation}
We recall that $\Qr^*$ is defined as in (\ref{Q0}) where $h=\Er_\Pr[H]$. With a slight change in notation, let $q$ be such that $\Pr[S^\Xb\ge q]\ge 1/h\ge \Pr[S^\Xb > q]$.  Then 
$\Qr^*[ \{S^\Xb<q\}]=0$ and  ${\Qr^*}[S^\Xb>K]=\Er_{\Pr}\[H\one_{\{S^\Xb>K\}}\]$ (since obviously $K\ge q$).

For $i\geq 1$ and $\Qr\in\Sc$ we have:
\begin{align*}
\Er_{\Qr}[X_i-Y_i]&=-\int_{S^\Xb<K} \alpha_i(K-S^\Xb)d\Qr +\int_{S^\Xb>K} \frac{X_i}{S^\Xb}(S^\Xb-K)\left(\frac{d\Qr}{d\Pr}\right)d\Pr \\
&\leq \alpha_i\rho(-(K-S^\Xb)^+) +\int_{S^\Xb>K} \frac{X_i}{S^\Xb}(S^\Xb-K)Hd\Pr\\
&\leq \alpha_i\Er_{\Qr^*}[(K-S^\Xb)^+] +\int_{S^\Xb>K} \frac{X_i}{S^\Xb}(S^\Xb-K)Hd\Pr\\
&= \Er_{\Qr^*}[X_i-Y_i]
\end{align*}
We have used the fact that  $\rho(-(K-S^\Xb)^+)=\Er_{\Qr^*}[(K-S^\Xb)^+]$, which is a consequence of the commonotonicity of $\rho^H$ (Corollary \ref{corcomp}).  From the above inequality we obtain (\ref{uH}), for $i\in\{1,...,N\}$. 
\bigskip

\noindent
We now take $\rho$  cohesive, commonotonic, $\Sc$ weakly compact and prove that $\rho=\rho^H$ for some $H\in L^1$, that is,  there is a random variable $H$ so that  for any random variable $\xi\in L^\infty$, we have:
 \begin{equation}\label{whattoprove}
\rho(\xi)=\Er_{\Qr^\xi}[-\xi]
\end{equation}
where $\Qr^\xi$ is given in  (\ref{Q0}).

Because $\rho$ is commonotonic, its values are determined by the expression for sets. This brings us to the following reduction: find a random variable $0\le H\in L^1$ such that for all  $B\in \Fc$
\begin{equation}\label{whattoprove1}
\rho(-\one_B)=\Er_\Pr\[H\one_B\]\wedge 1=\rho^H(-\one_B).
\end{equation}
The natural candidate is of course $H=\sup_{\Qr\in\Sc}\frac{d\Qr}{d\Pr}$ and we will show that this function indeed works.  As a result we then get $\Sc=\Sc^H$.  The $\sup$ has to be understood in the measure theoretic sense and the reader who is not familiar with this concept can find the information in a course on measure theory.  However the proof below gives sufficient information to overcome the difficulties.

Let us first take $A\in\Fc$  such 
\begin{equation}\label{r1}
\rho(-\one_A)\in(0,1)
\end{equation} 
We consider a partition of $A$, $\tau(A)=\{A_1,A_2\}\subset \Fc$, the risks $X_1:=\one_{A_1}$, $X_2:=\one_{A_2}$ and their sum $S^\Xb =X_1+X_2=\one_{A}$.  We let 
$$
K:=\rho\(-\sum_i X_i\)=\rho\(-\one_A\),
$$
 and $K\in(0,1)$ due to  (\ref{r1}). This in turn leads to $A=\{S^\Xb>K\}$ and $A^c=\{S^\Xb<K\}$. There is $\Qr^A\in\nabla\rho (-\one_A)$ with $\rho(-\one_A)=\Qr^A[A]$.  Hence $\rho(\one_{A^c})=\rho(1-\one_A)=-1+\rho(-\one_A)=-\Qr^A[A^c]$.  As $\rho$ is cohesive, there exist offsetting payoffs for  the risk $\Xb=(X_1,X_2,0,..,0)$. We now consider $\Yb\in\Ar^\Xb(K)$ standard payoffs offsetting $\Xb$, with $Y_i=0$ for $i>2$. The residual risks of the units $i=1,2$ are given by
$$
X_i-Y_i= -\alpha_iK\one_{A^c}+(1-K)\one_{A_i}.
$$ By commonotonicity of $\rho$ and of the random variables $(Y_i-X_i)^+$ and $-(Y_i-X_i)^-$ the following hold (for $i=1,2$):
\begin{align*}
\rho(Y_i-X_i)&=\rho\((Y_i-X_i)^+\) +\rho\(-(Y_i-X_i)^-\)\\
&= \alpha_iK \rho(\one_{A^c})+(1-K)\rho(-\one_{A_i})\\
&= -\alpha_iK\Qr^A[A^c] +(1-K)\rho(-\one_{A_i})
\end{align*}
while by Proposition \ref{propequiv} we also know that for  $\Qr^A\in\nabla\rho (-\one_A)$:
\begin{align}\label{eqQ1i}
\rho(Y_i-X_i)=-\alpha_i K\Qr^A(A^c)+(1-K)\Qr^A(A_i)=0.
\end{align}
Therefore (remember that $0<K<1$), 
$$
\rho\(-\one_{A_i}\)=\Qr^A(A_i),\;\forall A_i\in\tau(A).
$$
We emphasize that the probability $\Qr^A$ is chosen independently of the partition of the set $A$, $\tau(A)$. That means: if  $\Qr^A\in \nabla \rho\(-\one_A\)$ then:
\begin{align}\label{qa}
\rho(-\one_{B})&=\Qr^A(B),\;\forall B\subset A.
%,\\\label{qb}
%\rho(\one_{A^c})&=-\Qr^A(A^c),
\end{align}
From (\ref{qa}) we deduce that  for all probabilities measures $\Qr\in\Sc$
$$
\frac{d\Qr^A}{d\Pr}\one_A\geq \frac{d\Qr}{d\Pr}\one_A\quad \Pr\text{ a.s.},
$$ 
that is,  $\frac{d\Qr^A}{d\Pr}\one_A\ge H\one_A$, and because $\Qr^A\in\Sc$ we must get equality. That means that if (\ref{r1}) holds we have $\frac{d\Qr^A}{d\Pr}=H$ on the set $A$. In other words
$$
\rho(-\one_A)=\Er[H\one_A].
$$
For general sets $B\in\Fc$, we distinguish between several cases.
\begin{itemize}
\item[(a)]If $B\in\Fc$ is such that $\rho(-\one_B)\in(0,1)$, then $\rho(-\one_B)=\Er[H\one_{B}]$, as was proved above.
\item[(b)] If $B\in \Fc$ satisfies $\rho(-\one_B)=0$, then, $B$ is a null set for all probabilities in $\Sc$, consequently $H\one_B=0\;\Pr \;\as$ and in a trivial way $\rho(-\one_B)=\Er_\Pr[H\one_B]$.
\item[(c)]If $B\in \Fc$ satisfies $\rho(-\one_B)=1$ we will use the weak compactness and the property that the probability space is atomless.  There is a nondecreasing family of sets $A_t;0\le t\le 1$ such that $\Pr[A_t]=t\Pr[B]$ and $A_1=B$. The Lebesgue property (weak compactness) then shows that the function $t\rightarrow \rho(-\one_{A_t})$ is continuous.  It is nondecreasing, starts at $0$ and ends at $1$.  Therefore there is a unique number $s\le1$ such that for $t<s$, $\rho(-\one_{A_t})<1$ and for $t\ge s$, $\rho(-\one_{A_t})=1$.  For $t<s$ we have $\rho(-\one_{A_t})=\Er[H\one_{A_t}]$ and by continuity we  get
$1=\rho(-\one_{A_s})=\Er[H\one_{A_s}]$. Since $A_s\subset B$ we have $1=\rho(-\one_B)=\Er[H\one_B]\wedge 1$.  
\end{itemize} 
We thus have  verified (\ref{whattoprove1}). We must still show that $H\in L^1$. This is rather obvious since $\rho(-\one_{\{H>n\}})=\Er[H\one_{\{H>n\}}]\wedge 1$ and the left side tends to $0$ as $n\rightarrow \infty$. Hence eventually $\Er[H\one_{\{H>n\}}]<1<\infty$ and $H\in L^1$.
\end{proof}
\begin{cor} Suppose that the risk measure $\rho$ is commonotonic, satisfies the weak compactness property, is cohesive and law determined (rearrangement invariant). Then $\rho$ is  a tail expectation, i.e. there is a level $\alpha$ with $\rho=AV@R_{(\Pr,\alpha)}$.
\end{cor}
Indeed if $\rho(\xi)$ is determined by the distribution of $\xi$, then the set $\Sc$ must be rearrangement invariant, i.e. if $f\in \Sc$ and $g$ has the same distribution as $f$, then also $g\in \Sc$.  It is easy to see (for an atomless space) that this implies that the function $H$ constructed above, must be a constant.  This is nothing else than the characterisation of AV@R.
\begin{rem}  For a scenario set $\Sc$ we can introduce a space of random variables.  We define
$$
E=\{\xi\mid \text{ for all }\Qr\in\Sc: \Er_\Qr[|\xi|]<\infty\}.
$$
For elements $\xi\in E$ we have that also $\rho(-|\xi|)=\sup_{\Qr\in\Sc}\Er_\Qr[|\xi|]<\infty$ and this expression defines a norm for which $E$ becomes a Banach space.  Examples show that in general $L^\infty$ is not dense in this space.  The risk measure $\rho$ has a natural extension to $E$.  Indeed we can define $\rho(\xi)= \sup_{\Qr\in\Sc}\Er_\Qr[-\xi]$.  In case $0<H\in L^1(\Pr), \Er[H]\ge 1$, the scenario set $\Sc^H$ also defines such a space which in this case is easy to describe.  As for tail expectation we have the following inequalities
$$
\int |\xi| \, d\Hr\le \rho(-|\xi|)\le h\int |\xi| \, d\Hr\text{ where }h=\Er[H]\text{ and }d\Hr=\frac{H}{h}\,d\Pr.
$$\
It follows that the space $E$ is  nothing else but the space $L^1(\Hr)$, with an equivalent norm.
\end{rem}
\section{Cohesion for fixed aggregated liability}
Above, we analysed group cohesion when the class of possible liability vectors is $\Xc$.  We now suppose that the overall (or aggregated) group liability is fixed and equals $Z\in L^\infty$ so that the class of all possible liability vectors becomes: $\{\Xb\in\Xc\;|\; \sum_{i=1}^N X_i=Z\}$. We generalise the class of risk measures  from the previous section as follows. We define:
\begin{equation}\label{uLH}
\rho^{L,H}(\xi) :=\sup_{\Qr\in\Sc^{L,H}}\Er_{\Qr}[-\xi],
\end{equation}
with a scenario set: 
\begin{equation}\label{propers}
\Sc^{L,H}:=\left \{\Qr: L\leq \frac{d\Qr}{d\Pr}\leq H \quad \Pr\; a.s.\right \},
\end{equation}
and where $L,H$ are nonnegative random variables that satisfy $0\leq L\leq H$, with $\Er[H]<\infty$ and $\Er[L]>0$. From the calculations below it will turn out that this risk measure is also commonotonic. We shall denote $\ell:=\Er[L]$ and $h:=\Er[H]$. We suppose $\ell<1$ and $h>1$ in order to exclude the trivial case $\Sc^{L,H}=\Pr$.
We  introduce the following probability measures:
\begin{equation}
\frac{d\Hr}{d\Pr}:=\frac{H-L}{h-\ell}\text{ and } \frac{d\Lr}{d\Pr}:=\frac{L}{\ell}\one_{\{\ell>0\}}+\one_{\{\ell=0\}}.
\end{equation}

For a given random variable $\xi$, we define a probability measure $\Qr^\xi$ as follows:
\begin{equation}\label{Qxi}
\frac{d\Qr^\xi}{d\Pr}:=
\begin{cases} 
H &\text{ on } \{\xi < q(\xi) \}\\
c(\xi) H+(1-c(\xi))L &\text{ on } \{\xi =q(\xi) \}\\
L &\text{ on } \{\xi  >q(\xi)\}
\end{cases} 
\end{equation}
with $q(\xi)$ and $c(\xi)$ being  constants (derived from the distribution of $\xi$)  to ensure  that $\Er_{\Pr}\[\frac{d\Qr^\xi}{d\Pr}\]=1$. These constants can be computed as follows. Let us denote:
$$
F(\xi,x):=\Er_{\Pr}\[H\one_{\xi\leq  x}+L\one_{\xi>  x}\]=\ell+\Er_\Pr\[(H-L)\one_{\xi\leq x}\].
$$
The function $F(\xi,\cdot)$ is increasing, right continuous,  and satisfies $\lim_{x\to-\infty} F(\xi,x)= \ell\geq 0$ and $\lim_{x\to\infty} F(\xi,x)= h \geq 1$. We denote:
\begin{equation}\label{qxi}
q(\xi):=\inf\{x: F(\xi,x)\geq 1\}
\end{equation} and define $c(\xi)$ as   satisfying $c(\xi) F(q(\xi))+(1-c(\xi))F(q(\xi)-)=1$, that is:
\begin{equation}\label{cxi}
c(\xi)=
\begin{cases}
0 &  \text{if $F$ is continuous at $q(\xi)$} \\
\frac{1-F(q(\xi)-)}{F(q(\xi))-F(q(\xi)-)}&\text{otherwise.} 
\end{cases}
\end{equation}
We observe that indeed  $\Er_{\Pr}\[\frac{d\Qr^\xi}{d\Pr}\]=c(\xi) F(q(\xi))+(1-c(\xi))F(q(\xi)-) =1$ as required for $\Qr^\xi$ to be a probability measure.

\begin{prop}
 For any $\xi\in L^\infty$ we have the following alternative representations for $\rho^{L,H}$:

\begin{align*}
\rho^{L,H}(\xi)&=\ell\, \Er_\Lr[-\xi]+(1-\ell) AV@R_{(\Hr,\gamma)}(\xi).
\end{align*}
for $\gamma=(1-\ell)/(h-\ell))<1$ and
$$
\rho^{L,H}(\xi)= \Er_{\Qr^\xi}[-\xi],
$$ 
where $\Qr^\xi$ defined in (\ref{Qxi}).

\end{prop}

\begin{proof}We have by definition:
\begin{equation}
\rho^{L,H}(\xi)=\sup \left\{-\Er_\Pr[\varphi \xi ]\;|\; L\leq \varphi\leq H\;,\; \Er_\Pr[\varphi]=1\right\} 
\end{equation}

Using the transformation $\psi:=\frac{(\varphi -L)(h-\ell)}{(H-L)(1-\ell)}$ we obtain that $\{ L\leq \varphi\leq H\;,\; \Er_\Pr[\varphi]=1\}= \{ 0\leq \psi\leq \frac{h-\ell}{1-\ell}\;,\; \Er_{\Hr}[\psi]=1\} $ and:

\begin{align*}
\Er_\Pr[\varphi \xi ]&=\Er_\Pr[L \xi ]+\Er_\Pr[(\varphi-L) \xi ]\\
&=\ell\,\Er_\Pr\[ \xi\frac{d\Lr}{d\Pr} \]+(1-\ell)\Er_{\Pr}\[\psi \xi \frac{d\Hr}{d\Pr}\]\\
&=\ell \,\Er_\Lr [\xi ]+(1-\ell)\Er_{\Hr}[\psi \xi ].
\end{align*}

Therefore, an equivalent expression for $\rho^{L,H}$ is:
\begin{align*}
\rho^{L,H}(\xi)
&=\ell \,\Er_\Lr[-\xi]+ (1-\ell) \sup \left \{-\Er_{\Hr}[\varphi \xi ]\;|\; 0\leq \varphi\leq \frac{h-\ell}{1-\ell}\;,\; \Er_{\Hr}[\varphi ]=1 \right \}\\
&=\ell \,\Er_\Lr[\xi]+(1-\ell) AV@R_{(\Hr,\gamma)}(\xi),
\end{align*}where $AV@R_{(\Hr,\gamma)}(-\xi)$ is the average value at risk for $-\xi$, at the level $\gamma=(1-\ell)/(h-\ell))$ and under the probability $\Hr$.We notice that $\gamma<1$ ads $h>1$.  The optimiser probability for $AV@R$ is known to be $\Qr^\xi$:
$$
\frac{d\Qr^\xi}{d\Hr} =\frac{1}{\gamma}\(\one_{\xi<q}+c\one_{\xi=q}\),
$$ with $q$ a $\gamma$ quantile of $\xi$ under $\Hr$ and $c=0$ if $\Hr(\xi=q)=0$  and  $c=\(\gamma-\Hr(\xi<q)\)/\Hr(\xi=q)$ otherwise. 
Writing $\frac{d\Qr^\xi}{d\Hr}\times \frac{d\Hr}{d\Pr} $  we obtain the expression (\ref{Qxi}) with  the constants $q=q(\xi)$ and $c= c(\xi)$.

\end{proof}
\begin{prop} Consider the regulator's risk measure is $\rho= \rho^{L,H}$. Let $Z\in L^\infty_+$ with  $q(-Z)$ the corresponding constant, as defined in (\ref{qxi}). Suppose further that
$$
\rho^{L,H}(-Z)\geq -q(-Z),
$$
(for this inequality to hold it is sufficient that $AV@R_{(\Hr,\gamma)}(-Z)\leq \Er_\Lr[Z]$).
Then, for all risk vectors $\Xb$ satisfying $\sum_{i=1}^N  X_i=Z$, the following equality is satisfied
$$
\Kc(\Xb)=\rho(-Z).
$$
\end{prop}
\begin{proof}Let us  denote $K=\rho^{L,H}(-Z)$ and $\Qr^*$ be the probability in (\ref{Qxi}) with $\xi=-Z$. If $K\geq -q(-Z)$,  then $\{Z>K\}\subset\{Z> -q(-Z)\}=\{-Z< q(-Z)\}=\{\frac{d\Qr^*}{d\Pr}=H\}$.

For any $\Xb\in\Xc$ satisfying  $\sum_{i=1}^N X_i=Z$ and for any corresponding vector $\Yb\in\Ar^\Xb(K)$, we have:  $\forall i$ $\{X_i>Y_i\}\subset\{Z>K\}$. Therefore,  we obtain:
$$
\rho(-(Y_i-X_i)^-)= \Er_{\Qr^*}\[-(Y_i-X_i)^-\]=\Er_\Pr\[-(Y_i-X_i)^-H\]. 
$$
It follows that the condition (\ref{cnsoff}) is satisfied and, as $\rho^{L,H}$ is commonotonic, the equality  $\Kc(\Xb)=\rho(-Z)$ holds as an application of Proposition \ref{offsettingStandardC}.

It remains to prove the claim that if $AV@R_{(\Hr,\gamma)}(-Z)\leq \Er_\Lr[Z]$ then $\rho^{L,H}(-Z)\geq -q(-Z)$. 
We have that $-q(-Z)\leq AV@R_{\Hr,\gamma}(-Z)$ (this is always true). Then, the condition\\  $AV@R_{(\Hr,\gamma)}(-Z)\leq \Er_\Lr[Z]$ implies $AV@R_{(\Hr,\gamma)}(-Z) \leq \rho(-Z)$, and hence the claim is proved.

\end{proof}

\end{document}